\documentclass[11pt]{article}

\usepackage{amssymb}
\usepackage{amstext}
\usepackage{amsfonts}
\usepackage{amsmath}
\usepackage{graphicx}
\usepackage{fullpage}

\newtheorem{theorem}{Theorem}[section]    % Specify Theorem
\newtheorem{fact}[theorem]{Fact}    % Specify Fact
\newtheorem{definition}{Definition}[section] % Specify Definition
\newtheorem{claim}[theorem]{Claim}    % Specify Claim
\newtheorem{lemma}[theorem]{Lemma}    % Specify Lemma
\newcommand{\qed}{\hfill{$\rule{6pt}{6pt}$}} %Box at end of proof
\newenvironment{proof}{\noindent{\bf Proof:}}{\qed\\}
\newenvironment{proofof}[1]{\noindent{\bf Proof of #1:}}{\qed\\}

\numberwithin{equation}{section} 

\newcommand{\floor}[1]{{\lfloor #1 \rfloor}}

\newcommand{\defeq}{\stackrel{\mathrm{def}}{=}}

\newcommand{\expct}{{\mathbb E}}

\newcommand{\suppress}[1]{}
\newcommand{\comment}[1]{}

\newcommand{\eps}{\varepsilon}

\newcommand{\sR}{{{\mathsf R}}}
\newcommand{\sD}{{{\mathsf D}}}
\newcommand{\mcM}{{{\mathcal M}}}
\newcommand{\mcX}{{{\mathcal X}}}
\newcommand{\mcY}{{{\mathcal Y}}}
\newcommand{\mcZ}{{{\mathcal Z}}}
\newcommand{\mcP}{{{\mathcal P}}}
\newcommand{\ment}{{{\mathsf{ment}}}}
\newcommand{\crent}{{{\mathsf{crent}}}}

\newcommand{\err}{{{\mathsf{err}}}}
\newcommand{\pub}{{{\mathsf{pub}}}}
\newcommand{\ess}{{{\mathsf{ess}}}}
\newcommand{\disj}{{{\mathsf{disj}}}}

\title{ {\bf A strong direct product theorem for two-way public coin communication complexity  }}

\author{Rahul Jain\thanks{Centre for Quantum Technologies and Department of Computer Science  National University of Singapore. {\tt rahul@comp.nus.edu.sg}}
}

\begin{document}

\maketitle

\abstract{We show a direct product result for two-way public coin communication complexity of all relations in terms of a new complexity measure that we define. Our new measure is a generalization to non-product distributions of the two-way product subdistribution bound of J., Klauck and Nayak~\cite{JainKN08}, thereby our result implying their direct product result in terms of the two-way product subdistribution bound. 

We show that our new complexity measure gives tight lower bound for the set-disjointness problem, as a result we reproduce strong direct product result for this problem, which was previously shown by Klauck~\cite{Klauck10}.}

\section{Introduction}
\label{sec-introduction}

Let $f \subseteq \mcX \times \mcY \times \mcZ$ be a relation and $\eps >0$. Let Alice with input $x \in \mcX$, and Bob with input $y \in \mcY$, wish to compute a $z \in \mcZ$ such that $(x,y,z) \in f$. We consider the model of public coin two-way communication complexity in which Alice and Bob exchange messages possibly using pubic coins and at the end output $z$. Let $\sR^{2,\pub}_\eps(f)$ denote the communication of the best protocol $\mcP$ which achieves this with error at most $\eps$ (over the public coins) for any input $(x,y)$. Now suppose that Alice and Bob wish to compute $f$ simultaneously on $k$ inputs $(x_1,y_1), \ldots, (x_k, y_k)$ for some $k \geq 1$. They can achieve this by running $k$ independent copies of $\mcP$ in parallel . However in this case the overall success could be as low as $(1 - \eps)^k$. Strong direct product conjecture for $f$ states that this is roughly the best that Alice and Bob can do. We show a direct product result in terms of a new complexity measure, the $\eps$ error two-way conditional relative entropy bound of $f$, denoted $\crent^2_{\eps}(f)$, that we introduce. 
%\begin{theorem}
%\label{thm:main}
%Let $f \subseteq \mcX \times \mcY \times \mcZ$ be a relation. Let $k \geq 1$ be a natural number. Then,
%$$ \sR^{2,\pub}_{1 - 2^{-\Omega(k)}}(f^k)  \geq   \Omega( k \cdot \crent^2_{1/3}(f)  )  \enspace . $$
%\end{theorem}
Our measure $\crent^2_{\eps}(f)$ forms a lower bound on $\sR^{2,\pub}_{\eps}(f) $ and forms an upper bound on the two-way product subdistribution bound of J., Klauck, Nayak~\cite{JainKN08}, thereby implying their direct product result in terms of the two-way product subdistribution bound. 

As an application we reproduce the strong direct product result for the set disjointness problem, first shown by Klauck~\cite{Klauck10}. We show that our new complexity measure gives tight lower bound for the set-disjointness problem. This combined with the direct product in terms of the new complexity measure, implies strong direct product result for the set disjointness problem. 

There has been substantial prior work on the strong direct product question and the weaker direct sum and weak direct product questions in various models of communication complexity, e.g.~\cite{ImpagliazzoRW94, ParnafezRW97,  ChakrabartiSWY01, Shaltiel03, JainRS03a, KlauckSdW04, Klauck04, JainRS05a,  BeamePSW07, Gavinsky08, JainKN08, JainK09,  HarshaJMR09, BarakBCR10, BravermanR10, Klauck10}.

In the next section we provide some information theory and communication complexity preliminaries that we need. We refer the reader to the texts~\cite{CoverT91, KushilevitzN97} for good introductions to these topics respectively. In section~\ref{sec:dpt-2} we introduce our new bound and show the direct product result. In section~\ref{sec:disj} we show the application to set disjointness.

\section{Preliminaries}
\subsection*{Information theory}
\label{sec:inf}
Let $\mcX, \mcY$ be sets and $k$ be a natural number. Let $\mcX^k$ represent $\mcX \times \cdots \times \mcX$, $k$ times. Let $\mu$ be a distribution over $\mcX$ which we denote by $\mu \in \mcX$.  We use $\mu(x)$ to represent the probability of $x$ under $\mu$. The entropy of $\mu$ is defined as $S(\mu) = - \sum_{x \in \mcX} \mu(x) \log \mu(x)$. Let $X$ be a random variable distributed according to $\mu$ which we denote by $X \sim \mu$. We use the same symbol to represent a random variable and its distribution whenever it is clear from the context. For distributions $\mu, \mu_1 \in \mcX$, $\mu \otimes \mu_1$ represents the product distribution $(\mu \otimes \mu_1)(x) = \mu(x) \otimes \mu_1(x)$ and $\mu^k$ represents $\mu \otimes \cdots \otimes \mu$, $k$ times. The $\ell_1$ distance between distributions $\mu, \mu_1$ is defined as $||\mu - \mu_1||_1 = \frac{1}{2} \sum_{x \in \mcX} |\mu(x) - \mu_1(x)|$. Let $\lambda, \mu \in \mcX \times \mcY$.  We use $\mu(x | y)$ to represent $\mu(x,y)/\mu(y)$. When we say $XY \sim \mu$ we assume that $X \in \mcX$  and $Y \in \mcY$. We use $\mu_x$ and $Y_x$ to represent $Y |~ X =x$.  The conditional entropy of $Y$ given $X$, is defined as $S(Y|X) = \expct_{x \leftarrow X} S(Y_x)$. The relative entropy between $\lambda$ and $\mu$ is defined as $S(\lambda || \mu) = \sum_{x \in \mcX} \lambda(x) \log \frac{\lambda(x)}{\mu(x)}$. We use the following properties of relative entropy at many places without explicitly mentioning.
\begin{fact}
\label{fact:relprop}
\begin{enumerate}
\item Relative entropy is jointly convex in its arguments, that is for distributions $\lambda_1, \lambda_2,\mu_1, \mu_2$
$$ S(p\lambda_1 + (1-p)\lambda_2~ ||~ p\mu_1 + (1-p) \mu_2) \leq p\cdot S(\lambda_1 || \mu_1) + (1-p) \cdot S(\lambda_2 || \mu_2) \enspace .$$
\item Let $XY, X^1Y^1 \in \mcX \times \mcY$. Relative entropy satisfies the following chain rule, 
$$S(X Y || X^1 Y^1 ) = S(X || X^1) + \expct_{x \leftarrow X} S(Y_x || Y^1_x) \enspace .$$
This in-particular implies, using joint convexity of relative entropy,
\begin{align*}
S(X Y || X^1 \otimes Y^1 ) & = S(X || X^1)  +  \expct_{x \leftarrow X} S(Y_x || Y^1) \geq S(X || X^1)  + S(Y || Y^1)  \quad .
\end{align*}
\item For distributions $\lambda, \mu$ : $||\lambda - \mu ||_1 \leq \sqrt{S(\lambda || \mu)}$  and $S(\lambda || \mu) \geq 0$.
\end{enumerate}
\end{fact}
The relative min-entropy between $\lambda$ and $\mu$ is defined as $S_\infty(\lambda || \mu) = \max_{x \in \mcX} \log \frac{\lambda(x)}{\mu(x)}$. It is easily seen that $S(\lambda || \mu) \leq S_\infty(\lambda || \mu)$.
Let $X,Y,Z$ be random variables. The mutual information between $X$ and $Y$ is defined as 
$$I(X : Y) = S(X) + S(Y) - S(XY) = \expct_{x \leftarrow X} S(Y_x || Y) =  \expct_{y \leftarrow Y} S(X_y || X). $$ 
The conditional mutual information is defined as $I(X:Y |~ Z) = \expct_{z \leftarrow Z} I(X : Y |~ Z=z) $. Random variables $XYZ$ form a Markov chain $Z \leftarrow X \leftarrow Y$ iff $I(Y : Z  |~ X =x) = 0$ for each $x$ in the support of $X$.

\subsubsection*{Two-way communication complexity} 
Let $f \subseteq \mcX \times \mcY \times \mcZ$ be a relation. We only consider complete relations, that is for all $(x,y) \in \mcX \times \mcY$, there exists a $z \in \mcZ$ such that $(x,y,z) \in f$. 
In the two-way model of communication, Alice with input $x \in \mcX$ and Bob with input $y \in \mcY$, communicate at the end of which they are supposed to determine an answer $z$ such that $(x,y,z) \in f$. Let~$\eps > 0$ and let $\mu \in \mcX
\times \mcY$ be a distribution. We let $\sD_{\eps}^{2,\mu}(f)$ represent the two-way distributional communication complexity of $f$ under $\mu$
with expected error $\epsilon$, i.e., the communication of the best deterministic two-way protocol for $f$, with distributional
error  (average error over  the inputs) at most
$\eps$ under $\mu$.  Let $\sR^{2,\pub}_{\epsilon}(f)$ represent the public-coin two-way communication complexity of $f$ with worst case error
$\eps$, i.e., the communication of the best public-coin two-way protocol for $f$ with error for each input $(x,y)$ being at
most~$\eps$.  The following is a consequence of the  min-max theorem in game
theory~\cite[Theorem~3.20, page~36]{KushilevitzN97}.
\begin{lemma}[Yao principle]
\label{lem:yao} $\sR^{2,\pub}_{\epsilon}(f) = \max_{\mu}
\sD_{\epsilon}^{2,\mu}(f)$.
\end{lemma}

\section{A strong direct product theorem for two-way communication complexity}
\label{sec:dpt-2}

\subsection{New bounds}
\label{sec:newbounds-2}
Let $f \subseteq \mcX \times \mcY \times \mcZ$ be a relation, $\mu , \lambda \in \mcX \times \mcY$ be distributions and $\eps > 0$. Let $XY \sim \mu$ and $X_1Y_1 \sim \lambda$ be random variables. Let $S \subseteq \mcZ$.

\begin{definition}[Error of a distribution]
Error of  distribution $\mu$ with respect to $f$ and answer in $S$, denoted $\err_{f,S}(\mu)$, is defined as 
 $$\err_{f,S}(\mu) \defeq \min\{ \Pr_{(x,y) \leftarrow \mu}[(x,y,z) \notin f] ~|~ z \in S\} \enspace .$$
\end{definition}

\begin{definition}[Essentialness of an answer subset] Essentialness of answer in $S$ for $f$ with respect to distribution $\mu$, denoted $\ess^\mu(f,S)$, is defined as 
$$ \ess^\mu(f,S) \defeq 1 - \Pr_{(x,y) \leftarrow \mu} [\mbox{there exists $z \notin S$ such that $(x,y,z) \in f$}]  .$$
\end{definition}
For example $\ess^\mu(f,\mcZ) = 1$.
\begin{definition}[One-way distributions]
$\lambda$  is called one-way for $\mu$  with respect to $\mcX$, if for all $(x,y)$ in the support of $\lambda$ we have $\mu(y | x) = \lambda(y | x)$.
Similarly $\lambda$  is called one-way for $\mu$  with respect to $\mcY$, if for all $(x,y)$ in the support of $\lambda$ we have  $\mu(x | y) = \lambda(x | y)$.
\end{definition}

\begin{definition}[SM-like] 
$\lambda$ is  called SM-like (simultaneous-message-like) for $\mu$, if there is a
    distribution~$\theta$ on~$\mcX \times \mcY$ such that~$\theta$ is
    one-way for $\mu$ with respect to~$\mcX$ and~$\lambda$ is
    one-way for~$\theta$ with respect to~$\mcY$.
\end{definition}

\begin{definition}[Conditional relative entropy]
The $\mcY$-conditional relative entropy of $\lambda$ with respect to $\mu$, denoted $\crent^\mu_{\mcY}(\lambda)$, is defined as
$$ \crent^\mu_{\mcY} (\lambda) \defeq \expct_{y \leftarrow Y_1} S( (X_1)_y || X_y ) . $$
Similarly the $\mcX$-conditional relative entropy of $\lambda$ with respect to $\mu$, denoted $\crent^\mu_{\mcX}(\lambda)$, is defined as
$$ \crent^\mu_{\mcX} (\lambda) \defeq \expct_{x \leftarrow X_1} S( (Y_1)_x || Y_x) . $$
\end{definition}

\begin{definition}[Conditional relative entropy bound]
The  two-way $\eps$-error conditional relative entropy bound of $f$  with answer in $S$ with respect to distribution $\mu$, denoted $\crent^{2,\mu}_\eps(f,S)$, is defined as 
$$ \crent^{2,\mu}_\eps(f,S) \defeq \min \{ \crent^\mu_\mcX(\lambda) + \crent^\mu_\mcY(\lambda) ~ |~ \lambda \mbox{ is SM-like for $\mu$ and $\err_{f,S}(\lambda) \leq \eps$} \} \enspace. $$ 
The  two-way $\eps$-error conditional relative entropy bound of $f$, denoted $\crent^2(f)$, is defined as
$$ \crent^2_\eps(f) \defeq \max \{ \ess^\mu(f,S) \cdot \crent^{2,\mu}_\eps(f,S) ~|~ \mu \mbox{ is a distribution over } \mcX \times \mcY \mbox{ and } S \subseteq \mcZ\} \enspace. $$ 
\end{definition}

The following bound is analogous to a bound defined in~\cite{JainKN08} where it was referred to as the two-way subdistribution bound. We call it differently here for consistency of nomenclature with the other bounds. \cite{JainKN08} typically considered the cases where $S = \mcZ$ or $S$ is a singleton set.
\begin{definition}[Relative min entropy bound]
\label{def:relminbd}
The two-way $\eps$-error relative min entropy bound of $f$  with answer in $S$ with respect to distribution $\mu$, denoted $\ment^{2,\mu}_\eps(f,S)$, is defined as 
$$ \ment^{2,\mu}_\eps(f,S) \defeq \min \{ S_\infty(\lambda || \mu) |~ \lambda \mbox{ is SM-like for $\mu$ and $\err_{f,S}(\lambda) \leq \eps$} \} \enspace. $$ 
The  two-way $\eps$-error relative min entropy bound of $f$, denoted $\ment^2_\eps(f)$, is defined as
$$ \ment^2_\eps(f) \defeq \max \{ \ess^\mu(f,S) \cdot \ment^{2,\mu}_\eps(f,S) ~|~ \mu \mbox{ is a distribution over } \mcX \times \mcY \mbox{ and } S \subseteq \mcZ\} \enspace. $$ 
\end{definition}
The following is easily seen from definitions.
\begin{lemma} 
\label{lem:relations}
$$\crent^\mu_\mcX(\lambda) + \crent^\mu_\mcX(\lambda) \leq 2 \cdot S_\infty(\lambda || \mu)$$
 and hence 
$$\crent^{2,\mu}_\eps(f,S) \leq 2 \cdot  \ment^{2,\mu}_\eps(f,S) \quad \mbox{ and } \quad \crent^2_\eps(f) \leq 2 \cdot  \ment^2_\eps(f) .$$
\end{lemma}
It can be argued using the substate theorem~\cite{JainRS02} (proof skipped) that when $\mu$ is a product distribution then  $\ment^{2,\mu}_\eps(f,S) =  O(\crent^{2,\mu}_{\eps/2}(f,S))$. Hence our bound $\crent^{2}_\eps(f)$ is an upper bound on the product subdistribution bound of~\cite{JainKN08} (which is obtained when in Definition~\ref{def:relminbd} maximization is done only over product distributions $\mu$).

\subsection{Strong direct product}
\label{sec:dpt}

\noindent {\bf Notation:} Let $B$ be a set. For a random variable distributed in $B^k$, or a string in $B^k$, the portion corresponding to the $i$th coordinate is represented with subscript $i$. Also the portion except the $i$th coordinate is represented with subscript $-i$. Similarly portion corresponding to a subset $C \subseteq [k]$ is represented with subscript $C$. For joint random variables $MN$, we let $M_n$ to represent $M |~ (N=n)$ and also $M N |~ (N=n)$ and is clear from the context.

We start with the following theorem which we prove later. 
\begin{theorem}[Direct product in terms of $\ment$ and $\crent$]
\label{thm:dptment-2}
Let $f \subseteq \mcX \times \mcY \times \mcZ$ be a relation, $\mu \in \mcX \times \mcY$ be a distribution and $S \subseteq \mcZ$. Let $0 < \eps < 1/3$, $0 < 200 \delta < 1$ and $k$ be a natural number.  
Fix $z \in \mcZ^k$. Let the number of indices $i \in [k]$ with $z_{i} \in S$ be at least $\delta_1 k$ .  Then
$$ \ment^{2,\mu^k}_{1 - (1- \eps/2)^{\floor {\delta \delta_1 k}}}(f^k,\{z\})  \geq \delta \cdot \delta_1 \cdot k \cdot \crent^{2,\mu}_\eps(f,S) \enspace .$$
\end{theorem}
We now state and prove our main result.
\begin{theorem}[Direct product in terms of $\sD$ and $\crent$]
\label{thm:dpttwoway}
Let $f \subseteq \mcX \times \mcY \times \mcZ$ be a relation, $\mu \in \mcX \times \mcY$ be a distribution and $S \subseteq \mcZ$.  Let $0 < \eps < 1/3$  and $k$ be a natural number.    Let $\delta_2 = \ess^\mu(f,S)$. Let $0 < 200 \delta < \delta_2 $. 
Let $\delta' = 3(1- \eps/2)^{\floor {\delta \delta_2 k/2}} $. Then,
$$  \sD^{2,\mu^k}_{1 - \delta'}(f^k)   \geq \delta \cdot \delta_2 \cdot k \cdot \crent^{2,\mu}_\eps(f,S) - k\enspace .$$
%In other words, by maximizing over $\mu, S$ and using Lemma~\ref{lem:yao}.
%$$ \sR^{2,\pub}_{1 - 2^{-\Omega(k)}}(f^k)  \geq   \Omega(k \cdot \crent^{2}_{1/3}(f)   )  \enspace . $$
\end{theorem}
\begin{proof}
Let $\crent_{2,\eps}^{\mu}(f,S) = c$. For input $(x,y) \in \mcX^k \times \mcY^k$, let $b(x,y)$ be the number of indices $i$ in $[k]$ for which  there exists $z_i \notin S$ such that $(x_i,y_i,z_i) \in f$. Let 
$$B = \{(x,y) \in \mcX^k \times \mcY^k |~ b(x,y) \geq (1 - \delta_2/2) k  \}. $$
By Chernoff's inequality we get,
$$ \Pr_{(x,y) \leftarrow \mu^k}[(x,y) \in B] \leq \exp(-\delta_2^2 k/2) .$$
Let $\mcP$ be a protocol for $f^k$ with inputs $XY \sim \mu^k$ with communication at most $d = (k c \delta \delta_2/2) - k$ bits.  Let $M \in \mcM$ represent the message transcript of $\mcP$.  Let 
$$B_M = \{ m \in \mcM  | ~ \Pr[(XY)_m \in B] \geq \exp(-\delta_2^2 k/4) \} .$$ Then $ \Pr[M \in B_M] \leq \exp(-\delta_2^2 k/4) .$ Let 
$$B^1_M = \{ m \in \mcM  | ~ \Pr[M = m] \leq 2^{-d - k} \} .$$ Then 
$ \Pr[M \in B^1_M] \leq 2^{- k} .$ Fix $ m \notin B_M \cup B^1_M$. Let $z_m$ be the output of $\mcP$ when $M=m$. Let $b(z_m)$ be the number of indices $i$ such that $z_{m,i} \notin S$. If $b(z_m) \geq 1 -\delta_2 k/2$ then success of $\mcP$ when $M=m$ is at most $\exp(-\delta_2^2 k/4) \leq (1- \eps/2)^{\floor {\delta \delta_2 k/2}}$. If $b(z_m) < 1 -\delta_2 k/2$ then from Theorem~\ref{thm:dptment-2} (by setting $z = z_m$ and $\delta_1 = \delta_2/2$), success of $\mcP$ when $M=m$ is at most $(1- \eps/2)^{\floor {\delta \delta_2 k/2}}$. Therefore overall success of $\mcP$ is at most
\begin{align*}
\delta' & = 2^{- k} +  \exp(-\delta_2^2 k/4)  + (1-2^{-k} -  \exp(-\delta_2^2 k/4) (1- \eps/2)^{\floor {\delta \delta_2 k/2}} \\
& \leq  3(1- \eps/2)^{\floor {\delta \delta_2 k/2}} .
\end{align*}

\end{proof}

\begin{proofof}{Theorem~\ref{thm:dptment-2}}
Let $c = \crent^{2,\mu}_\eps(f,S) $. Let $\lambda \in \mcX^k \times \mcY^k$ be a distribution which is SM-like for $\mu^k$ and with $S_\infty (\lambda || \mu^k) < \delta \delta_1 c k$. We show that $\err_{f^k, \{z\}}(\lambda) \geq 1 - (1 - \eps/2)^{\floor{ \delta \delta_1 k  }}$. This shows the desired. 

Let $X Y \sim \lambda$.  For a coordinate $i$, let the binary random variable $T_i \in \{0,1\}$, correlated with $XY$, denote success in the $i$th coordinate.  That is $T_i = 1$ iff $XY = (x,y)$ such that $(x_i, y_i, z_i) \in f$. 
We make the following claim which we prove later. Let $k' = \floor {\delta \delta_1 k  }$.
\begin{claim}
\label{claim:succ}
There exists $k'$ distinct coordinates $i_1, \ldots ,i_{k'}$  such that $\Pr[T_{i_1} = 1] \leq 1 - \eps/2$ and for each  $r < k' $,
\begin{enumerate}
\item either $\Pr[T_{i_1} \times  T_{i_2} \times \cdots \times T_{i_r} = 1 ] \leq (1 - \eps/2)^{k'}$,
\item or  $\Pr[T_{i_{r+1}} = 1 |~ (T_{i_1} \times  T_{i_2} \times \cdots \times T_{i_r} = 1)] \leq 1 - \eps/2 $. 
\end{enumerate}
\end{claim}
This shows that the overall success is 
$$ \Pr[T_1 \times T_2 \times \cdots \times T_k  = 1] \leq \Pr[T_{i_1} \times T_{i_2} \times \cdots \times T_{i_{k'}}  = 1] \leq (1-\eps/2)^{k'} \enspace .$$
\end{proofof}

\begin{proofof}{Claim~\ref{claim:succ}} 
Let us say we have identified $r < k'$ coordinates $i_1, \ldots i_r$ .  Let $C = \{ i_1, i_2, \ldots, i_r\}$. Let $T = T_{1} \times T_{2} \times \cdots \times T_{r}$ . If $\Pr[T=1] \leq (1- \eps/2)^{k'}$ then we will be done. So assume that $ \Pr[T=1]  > (1 - \eps/2)^{k'} \geq 2^{-\delta \delta_1 k } $.  Let $X' Y' \sim \mu$.
Let $X^1 Y^1  = (XY | ~ T = 1)$. Let  $D$ be uniformly distributed  in $\{0,1\}^k$ and independent of $X^1Y^1$. Let $U_i = X^1_i$ if $D_i=0$ and $U_i = Y^1_i$ if $D_i=1$.  Let $U = U_1 \ldots U_k$. Below for any random variable $\tilde{X}\tilde{Y}$, we let $\tilde{X}\tilde{Y}_{d,u}$, represent the random variable obtained by appropriate conditioning on $\tilde{X}\tilde{Y}$: for all $i$, $\tilde{X}_i = u_i$ if $d_i=0$ otherwise $\tilde{Y}_i = u_i$ if $d=1$ . Let $I$ be the set of indices $i$ such that $z_i \in S$. 
 Consider,
\begin{align}
 \delta \delta_1 k + \delta \delta_1 c k   &>  S_\infty(X^1Y^1  || XY)  + S_\infty(XY || (X'Y')^{\otimes k} ) \nonumber \\
& \geq  S_\infty(X^1Y^1 || (X'Y') ^{\otimes k} )  \geq S(X^1Y^1 || (X'Y') ^{\otimes k} ) = \expct_{d \leftarrow D} S(X^1Y^1 || (X'Y') ^{\otimes k} ) \nonumber\\
& \geq \expct_{(d,u,x_C,y_C) \leftarrow (D U X^1_C Y^1_C)} S((X^1Y^1)_{d,u,x_C,y_C}  || ((X'Y')^{\otimes k})_{d,u,x_C,y_C} ) \nonumber \\
& \geq  \expct_{(d,u,x_C,y_C) \leftarrow (D U X^1_C Y^1_C)} S(X^1_{d,u,x_C,y_C}  || X'_{d_1,u_1,x_C,y_C} \otimes \ldots \otimes X'_{d_k,u_k,x_C,y_C}  )  \nonumber \\
& \geq  \expct_{(d,u,x_C,y_C) \leftarrow (D U X^1_C Y^1_C)} \sum_{i \notin C, i \in I} S((X^1_{d,u,x_C,y_C})_i  || X'_{d_i,u_i}  ) \nonumber \\
 &= \sum_{i \notin C, i \in I} \expct_{(d,u,x_C,y_C) \leftarrow (D U X^1_C Y^1_C)}  S((X^1_{d,u,x_C,y_C})_i  || X'_{d_i,u_i}  )   \enspace . \label{eq:1}
\end{align}
Similarly,
\begin{align}
\delta \delta_1 k + \delta \delta_1 c k  >   \sum_{i \notin C, i \in I} \expct_{(d,u,x_C,y_C) \leftarrow (D U X^1_C Y^1_C)}  S((Y^1_{d,u,x_C,y_C})_i  || Y'_{d_i,u_i}  )   \enspace . \label{eq:2}
\end{align}
From Eq.~\ref{eq:1} and Eq.~\ref{eq:2} and using Markov's inequality we get a coordinate $j $ outside of $C$ but in $I$ such that 
\begin{enumerate}
\item $ \expct_{(d,u,x_C,y_C) \leftarrow (D U  X^1_C Y^1_C)}  S((X^1_{d,u,x_C,y_C})_j  || X'_{d_j,u_j}  )  \leq \frac{2 \delta (c+1)}{(1-\delta )} \leq 4 \delta  c  ,$ and
\item $ \expct_{(d,u,x_C,y_C) \leftarrow (D U  X^1_C Y^1_C)} S((Y^1_{d,u,x_C,y_C})_j  || Y'_{d_j,u_j} ) \leq \frac{2 \delta (c+1)}{(1-\delta ) } \leq 4 \delta c$.
\end{enumerate}
Therefore,
\begin{align*}
4 \delta c  &\geq \expct_{(d,u,x_C,y_C) \leftarrow (D U  X^1_C Y^1_C)}  S((X^1_{d,u,x_C,y_C})_j  || X'_{d_j,u_j}  ) \\
& = \expct_{(d_{-j},u_{-j},x_C,y_C) \leftarrow (D_{-j} U_{-j}  X^1_C Y^1_C)} \expct_{(d_j,u_j) \leftarrow (D_j U_j)  | ~(D_{-j} U_{-j}  X^1_C Y^1_C) = (d_{-j},u_{-j},x_C,y_C)  }   S((X^1_{d,u,x_C,y_C})_j  || X'_{d_j,u_j}  )    .
\end{align*}
 And,
\begin{align*}
4 \delta c &\geq \expct_{(d,u,x_C,y_C) \leftarrow (D U X^1_C Y^1_C)}   S((Y^1_{d,u,x_C,y_C})_j  || Y'_{d_j,u_j} ) \\
& = \expct_{(d_{-j},u_{-j},x_C,y_C) \leftarrow (D_{-j} U_{-j} X^1_C Y^1_C)} \expct_{(d_j,u_j) \leftarrow (D_j U_j )  | ~(D_{-j} U_{-j} X^1_C Y^1_C) = (d_{-j},u_{-j},x_C,y_C)  }  S((Y^1_{d,u,x_C,y_C})_j  || Y'_{d_j,u_j} )  .
\end{align*}
Now using Markov's inequality, there exists set $G_1 $ with $\Pr[D_{-j} U_{-j} X^1_C Y^1_C \in G_1] \geq 1 - 0.2 $,  such that for all $(d_{-j},u_{-j},x_C,y_C) \in G_1$,
\begin{enumerate}
\item \label{it:1} $ \expct_{(d_j,u_j) \leftarrow (D_j U_j )  | ~(D_{-j} U_{-j} X^1_C Y^1_C) = (d_{-j},u_{-j},x_C,y_C)  }     S((X^1_{d,u,x_C,y_C})_j  || X'_{d_j,u_j}  )    \leq 40 \delta c$, \quad and
\item \label{it:2} $ \expct_{(d_j,u_j) \leftarrow (D_j U_j )  | ~(D_{-j} U_{-j} X^1_C Y^1_C) = (d_{-j},u_{-j},x_C,y_C)  }     S((Y^1_{d,u,x_C,y_C})_j  || Y'_{d_j,u_j} ) \leq 40 \delta c$.
\end{enumerate}
Fix $(d_{-j},u_{-j},x_C,y_C) \in G_1$. Conditioning on $D_j=1$ (which happens with probability $1/2$) in inequality~\ref{it:1}. above we get,
\begin{equation}
\expct_{y_j \leftarrow Y^1_j  | (D_{-j} U_{-j} X^1_C Y^1_C) = (d_{-j},u_{-j},x_C,y_C)   }   S((X^1_{d_{-j},u_{-j},y_j,x_C,y_C})_j  || X'_{y_j}  )    \leq 80 \delta c .
\label{eq:3-2}
\end{equation}
Conditioning on $D_j=0$ (which happens with probability $1/2$) in inequality~\ref{it:2}. above we get,
\begin{equation}
\expct_{x_j \leftarrow X^1_j  | (D_{-j} U_{-j} X^1_C Y^1_C) = (d_{-j},u_{-j},x_C,y_C)   }  S((Y^1_{d_{-j},u_{-j},x_j,x_C,y_C})_j  || Y'_{x_j} )  \leq 80 \delta  c. \label{eq:4-2}
\end{equation}
Let $X^2Y^2 = ((X^1Y^1)_{d_{-j},u_{-j},x_C,y_C})_j$. Note that $X^2Y^2$ is SM-like for $\mu$. 
From Eq.~\ref{eq:3-2}  and Eq.~\ref{eq:4-2} we get that 
$$ \crent^{\mu}_\mcX(X^2Y^2) + \crent^{\mu}_\mcY(X^2Y^2) \leq c .$$
Hence,  
\begin{align*}
\err_f(((X^1Y^1)_{d_{-j},u_{-j},x_C,y_C})_j) \geq  \eps.
\end{align*}
This implies,
\begin{align*}
\Pr[T_j = 1|~ (1,d_{-j},u_{-j},x_C,y_C) = (T  D_{-j} U_{-j} X_C Y_C)] & \leq 1 - \eps .
\end{align*}
Therefore overall 
$$\Pr[T_j  = 1 | ~ (T=1)] \leq 0.8(1 -  \eps)  + 0.2  \leq 1 - \eps/2 .$$
\end{proofof}

\section{Strong direct product for set disjointness}
\label{sec:disj}
For a string $x \in \{0,1\}^n$ we let $x$ also represent the subset of $[n]$ for which $x$ is the characteristic vector. 
The set disjointness function $\disj : \{0,1\}^n \times \{0,1\}^n \rightarrow \{0,1\}$ is defined as $\disj(x,y) = 1$ iff the subsets $x$ and $y$ do not intersect. 
\begin{theorem}[Strong Direct product for set disjointness]
Let $k$ be a positive integer. Then $ \sR^{2,\pub}_{1 - 2^{-\Omega(k)}}(\disj^k) = \Omega( k \cdot n )$.
\end{theorem}
\begin{proof}
Let $n = 4l -1$ (for some integer $l$). Let $T = (T_1,T_2, I)$ be a uniformly random partition of $[n]$ into three disjoint sets such that $|T_1| = |T_2|  = 2l-1$ and $|I| =1$. Conditioned on $T = t = (t_1,t_2,\{i\})$, let $X$ be a uniformly random subset of $t_1 \cup \{i\}$ and $Y$ be a uniformly random subset of $t_2 \cup \{i\}$. 
Note that $X \leftrightarrow T \leftrightarrow Y$ is a Markov chain. We show,
\begin{lemma}
\label{lem:crentlarge}
$\crent^{2,XY}_{1/70}(\disj, \{1\}) = \Omega(n) $.
\end{lemma}
It is easily seen that $\ess^{XY}(\disj,\{1\})  = 0.75$. Therefore using Theorem~\ref{thm:dpttwoway} and Lemma~\ref{lem:yao} we have,
$$ \sR^{2,\pub}_{1 - 2^{-\Omega(k)}}(\disj^k) = \Omega( k \cdot n ). $$
\end{proof}
\begin{proofof}{Lemma~\ref{lem:crentlarge}}
Our proof follows on similar lines as the proof of Razborov showing linear lower bound on the rectangle bound for set-disjointness (see e.g.~\cite{KushilevitzN97}, Lemma 4.49). However there are differences since we are lower bounding a weaker quantity.

Let $\delta = 1/(200)^2$. Let $X' Y'$ be such that $\crent^{XY}_\mcX(X'Y') + \crent^{XY}_\mcY(X'Y')  \leq \delta n $ and $X'Y'$ is SM-like for $XY$. 
We will show that $\err_{\disj,\{1\}}(X'Y') = \Pr[\disj(X'Y') = 0 ] \geq 1/70 $.  This will show the desired. We assume that $ \Pr[\disj(X'Y') = 1 ] \geq 0.5$ otherwise we are done already.
Let $A,B \in \{0,1\}$ be binary random variables such that $A \leftrightarrow X \leftrightarrow Y \leftrightarrow B$ and $X'Y' = (XY | ~A=B=1) $.
\begin{claim}
\label{claim:good1}
\begin{align*}
1. \quad \lefteqn{ \Pr[A=B=1 ,  \disj(XY) = 0] } \\
& = \frac{1}{4} \expct_{t = (t_1,t_2,\{i\}) \leftarrow T} \Pr[A=1 |~ T=t, X_i = 1] \Pr[B=1 |~ T=t, Y_i = 1] .
\end{align*}
\begin{align*}
2. \quad \lefteqn{ \Pr[A=B=1 , \disj(XY) = 1] } \\
& = \frac{3}{4} \expct_{t = (t_1,t_2,\{i\}) \leftarrow T} \Pr[A=1 |~ T=t, X_i = 0] \Pr[B=1 |~ T=t, Y_i = 0] .
\end{align*}
\end{claim}

\begin{proof}
We first show part 1. 
\begin{align*}
\lefteqn{\Pr[A=B=1 ,   \disj(XY) = 0]  = \Pr[A=B=1 ,   X_I = Y_I  = 1] }\\
& =  \expct_{t= (t_1,t_2,\{i\}) \leftarrow T} \Pr[A=B=1,  X_i =Y_i =1  | ~T=t] \\
& =   \expct_{t= (t_1,t_2,\{i\}) \leftarrow T}  \Pr[ X_i =  Y_i=1 |~ T=t] \Pr[A=B=1  |~ T=t, X_i =  Y_i=1]  \\ 
& =   \frac{1}{4} \expct_{t= (t_1,t_2,\{i\}) \leftarrow T}  \Pr[A=B=1  |~ T=t, X_i =  Y_i=1] \\
& =   \frac{1}{4} \expct_{t= (t_1,t_2,\{i\}) \leftarrow T} \Pr[A=1  |~ T=t, X_i =1]  \Pr[B=1  |~ T=t,  Y_i=1]  .
\end{align*}

Now we show part 2. Note that the distribution of $ (XY |~ \disj(X,Y) = 1) $ is identical to the distribution of $(XY | ~X_I = Y_I = 0)$ (both being uniform distribution on disjoint $x,y$ such that $|x|=|y|=l$).  Also $\Pr[\disj(XY)=  1] = 3 \Pr[X_I = Y_I = 0 ] $. Therefore, 
\begin{align*}
\lefteqn{\Pr[A=B=1 ,  \disj(XY)=  1] = \Pr[\disj(XY)=  1] \Pr[A=B=1 |~  \disj(XY)=  1] }\\
& = 3 \Pr[X_I = Y_I = 0 ] \Pr[A=B=1 |~ X_I = Y_I = 0] = 3 \Pr[A=B=1 , X_I = Y_I = 0 ] \\
& =  3 \expct_{t= (t_1,t_2,\{i\}) \leftarrow T} \Pr[A=B=1 , X_i =0, Y_i =0 |~  T=t  ]  \\
& =   3\expct_{t= (t_1,t_2,\{i\}) \leftarrow T} \Pr[ X_i =0, Y_i =0 |~  T=t  ] \Pr[A=B=1 |~  T=t, X_i =0, Y_i =0   ]  \\
& =   \frac{3}{4} \expct_{t= (t_1,t_2,\{i\}) \leftarrow T} \Pr[A=B=1 |~  T=t, X_i =0, Y_i =0   ]  \\
& =   \frac{3}{4} \expct_{t= (t_1,t_2,\{i\}) \leftarrow T} \Pr[A=1 |~  T=t, X_i =0  ]  \Pr[B=1 |~  T=t,  Y_i =0   ]   .
\end{align*}
\end{proof}

\begin{claim}
\label{claim:good}
Let  $B^1_x = \{ t_2 | ~ S(X'_{t_2} || X_{t_2}) > 100 \delta n\} , \quad B^1_y = \{ t_1 | ~ S(Y'_{t_1} || Y_{t_1}) > 100 \delta n\} .$
$$B^2_x = \{ t | ~  \Pr[A=1 |~ X_i=1, T=t] < \frac{1}{3}  \Pr[A=1 |~ X_i=0, T=t] \} .$$
$$B^2_y = \{ t | ~  \Pr[B=1 |~ Y_i=1, T=t] < \frac{1}{3}  \Pr[B=1 |~ Y_i=0, T=t] \} .$$ 
\begin{enumerate}
\item $\Pr[A =B = 1, T_2 \in B^1_x] < \frac{1}{100} \Pr[A =B = 1] $.
\item $\Pr[A =B = 1, T_1 \in B^1_y] < \frac{1}{100}  \Pr[A =B = 1]$.
\item Let $t_2 \notin B^1_x$, then 
$ \Pr[T \in B^2_x |~ T_2 = t_2] < \frac{1}{100} .$
\item Let $t_1 \notin B^1_y$, then 
$ \Pr[T \in B^2_y |~ T_1 = t_1] < \frac{1}{100} .$
\end{enumerate}
\end{claim}

\begin{proof}
We show the proof of part 1. and part 2. follows similarly. Let $T' = (T |~A=B=1)$. Note that $X' \leftrightarrow T' \leftrightarrow Y'$ is a Markov chain.  Also for every $(x,y) :~$ $(T |~ XY = (x,y) )$ is identically distributed as $(T' |~X'Y'=(x,y))$. 
Consider,
\begin{align*}
\delta n & \geq \expct_{y \leftarrow Y'} S(X'_y || X_y)  =  \expct_{y \leftarrow Y'} S((X'T')_y || (XT)_y)  \\
&   \geq \expct_{(y,t) \leftarrow (Y'T')} S(X'_{y,t} || X_{y,t}) = \expct_{t \leftarrow T'} S(X'_{t} || X_{t}) = \expct_{t_2 \leftarrow T'_2}  S(X'_{t_2} || X_{t_2}) .
\end{align*}
Therefore using Markov's inequality,
\begin{align*}
\frac{1}{100} & >  \Pr[T_2' \in B^1_x]  = \Pr[T_2 \in B^1_x |~ A=B=1] =  \frac{\Pr[T_2 \in B^1_x , A=B=1]}{\Pr[ A=B=1]}.
\end{align*}

We show the proof of part 3. and part 4. follows similarly.  Fix $t_2 \notin B^1_x$. Then,
$$100 \delta n  \geq S(X'_{t_2} || X_{t_2}) \geq  \sum_{i \notin t_2}  S((X'_{t_2})_i || (X_{t_2})_i) .$$
Let $R = \{ i \notin t_2 |~ S((X'_{t_2})_i || (X_{t_2})_i) > 0.01 \}.$ From above  $\frac{|R| }{2l} <  \frac{1}{100} $. For $i \notin R \cup t_2$, 
\begin{align*}
& S((X'_{t_2})_i || (X_{t_2})_i) \leq 0.01 \quad \Rightarrow \quad || (X'_{t_2})_i - (X_{t_2})_i||_1 \leq \sqrt{0.01} = 0.1 \\
 & \Rightarrow  \Pr[(X'_{t_2})_i =1] \geq 0.4 \geq \frac{1}{3} \Pr[(X'_{t_2})_i =0]  \quad \mbox{(since $ \Pr[(X_{t_2})_i =1] = 0.5$)} \\
 & \Rightarrow   \Pr[X_i =1 | ~ T_2 = t_2 , A=1] \geq \frac{1}{3}   \Pr[X_i =0 | ~ T_2 = t_2, A=1] \\
 & \Rightarrow  \frac{\Pr[A=1|~ T_2 = t_2]}{\Pr[X_i =1 | ~ T_2 = t_2 ] } \Pr[X_i =1 | ~ T_2 = t_2 , A=1] \geq \frac{1}{3} \frac{\Pr[A=1|~ T_2 = t_2]}{ \Pr[X_i =0 | ~ T_2 = t_2]}   \Pr[X_i =0 | ~ T_2 = t_2, A=1] \\
 & \Rightarrow  \Pr[A=1 | ~ X_i =1 , T_2 = t_2 ] \geq \frac{1}{3} \Pr[A=1 | ~ X_i =0 , T_2 = t_2] .
\end{align*}
Therefore $i \notin R \cup t_2$ implies $t = (t_1,t_2,\{i\}) \notin B^2_x $. Therefore,
$$ \Pr[T \in B^2_x |~ T_2 = t_2] \leq  \Pr[i \in R |~ T_2 = t_2]  = \frac{|R| }{2l} <  \frac{1}{100}  .$$
\end{proof}

\begin{claim}
\label{claim:good2}
\begin{enumerate}
\item  Let $Bad^1_x = 1$ iff  $T_2 \in B^1_x$ otherwise $0$. Then 
\begin{align*}
\lefteqn{\expct_{t = (t_1,t_2,\{i\}) \leftarrow T}  \Pr[A=1 | ~  X_i = 0 , T=t]  \Pr[B=1 | ~  Y_i = 0 , T=t]  Bad^1_x } \\
&  \leq    \frac{6}{100}  \expct_{t = (t_1,t_2,\{i\}) \leftarrow T}  \Pr[A=1 | ~ X_i=0,  T=t] \Pr[B=1 | ~  Y_i = 0 , T=t]   .
\end{align*}
\item  Let $Bad^1_y = 1$ iff  $T_1 \in B^1_y$ otherwise $0$. Then 
\begin{align*}
\lefteqn{\expct_{t = (t_1,t_2,\{i\}) \leftarrow T}  \Pr[A=1 | ~  X_i = 0 , T=t]  \Pr[B=1 | ~  Y_i = 0 , T=t]  Bad^1_y } \\
&  \leq    \frac{6}{100}  \expct_{t = (t_1,t_2,\{i\}) \leftarrow T}  \Pr[A=1 | ~ X_i=0,  T=t] \Pr[B=1 | ~  Y_i = 0 , T=t]   .
\end{align*}
\item Fix $t_2 \notin B^1_x$. Let $T_{t_2} = (T|~T_2=t_2)$. Let $Bad^2_x = 1$ iff  $T \in B^2_x$ otherwise $0$. Then
\begin{align*}
\lefteqn{\expct_{t = (t_1,t_2,\{i\}) \leftarrow T_{t_2}}  \Pr[A=1 | ~  X_i = 0 , T=t]  \Pr[B=1 | ~  Y_i = 0 , T=t]  Bad^2_x } \\
&  \leq    \frac{2}{100}  \expct_{t = (t_1,t_2,\{i\}) \leftarrow T_{t_2}}  \Pr[A=1 | ~ X_i=0,  T=t] \Pr[B=1 | ~  Y_i = 0 , T=t]   .
\end{align*}
\item Fix $t_1 \notin B^1_y$. Let $T_{t_1} = (T|~T_1=t_1)$. Let $Bad^2_y = 1$ iff  $T \in B^2_y$ otherwise $0$. Then
\begin{align*}
\lefteqn{\expct_{t = (t_1,t_2,\{i\}) \leftarrow T_{t_1}}  \Pr[A=1 | ~  X_i = 0 , T=t]  \Pr[B=1 | ~  Y_i = 0 , T=t]  Bad^2_y } \\
&  \leq    \frac{2}{100}  \expct_{t = (t_1,t_2,\{i\}) \leftarrow T_{t_1}}  \Pr[A=1 | ~ X_i=0,  T=t] \Pr[B=1 | ~  Y_i = 0 , T=t]   .
\end{align*}
\end{enumerate}
\end{claim}

\begin{proof} We show part 1. and part 2. follows similarly.  Note that for all $t$,
\begin{align*}
\Pr[A=1 | ~ T=t] & = \Pr[X_i=0|~T=t] \Pr[A=1 | ~  X_i = 0 , T=t]  \\
& \quad + \Pr[X_i=1|~T=t] \Pr[A=1 | ~  X_i = 1 , T=t]   .
\end{align*} 
Hence $ \Pr[A=1 | ~ T=t] \geq \frac{1}{2} \Pr[A=1 | ~  X_i = 0 , T=t] $.  Similarly $ \Pr[B=1 | ~ T=t] \geq \frac{1}{2} \Pr[B=1 | ~  Y_i = 0 , T=t] $. Consider,
\begin{align*}
\lefteqn{\expct_{t = (t_1,t_2,\{i\}) \leftarrow T}  \Pr[A=1 | ~  X_i = 0 , T=t]  \Pr[B=1 | ~  Y_i = 0 , T=t]  Bad^1_x } \\
& \leq  4 \expct_{t = (t_1,t_2,\{i\}) \leftarrow T}  \Pr[A=1 | ~   T=t] \Pr[B=1 | ~ T=t]  Bad^1_x  \\
& =  4 \expct_{t = (t_1,t_2,\{i\}) \leftarrow T}  \Pr[A= B=1 | ~ T=t]  Bad^1_x  \\
&  = 4 \Pr[A=B=1, T_2 \in B^1_x] \\
& \leq \frac{4}{100}   \Pr[A=B=1] \quad \mbox{(from Claim~\ref{claim:good})} \\
 & \leq \frac{8}{100}   \Pr[A=B=1 , \disj(XY) = 1]  \quad \mbox{(since $ \Pr[\disj(X'Y') = 1 ] \geq 0.5$)} \\
 & = \frac{6}{100} \expct_{t = (t_1,t_2,\{i\}) \leftarrow T}   \Pr[A=1 |~ T=t, X_i=0]  \Pr[B=1 |~ T=t, Y_i=0]  \quad  \mbox{(from Claim~\ref{claim:good1})}
\end{align*}

\vspace{0.1in}

We show part 3. and part 4. follows similarly. Note that :
\begin{enumerate}
\item $ \Pr[B=1 | ~  Y_i = 0 , T=(t_1,t_2,\{i\})] $ is independent of $i$ for fixed $t_2$. Let us call it $c(t_2)$.
\item $ \Pr[A=1 | ~  T=(t_1,t_2,\{i\})] $ is independent of $i$ for fixed $t_2$. Let us call it $r(t_2)$.
\item Distribution of $(X |~T_2=t_2)$ is identical to the distribution $(X|~T_2=t_2, X_I=0)$. Hence $  \expct_{t = (t_1,t_2,\{i\}) \leftarrow T_{t_2}}  \Pr[A=1 | ~   T=t]   = \expct_{t = (t_1,t_2,\{i\}) \leftarrow T_{t_2}}  \Pr[A=1 | ~ X_i=0,  T=t]  $.
\end{enumerate}
Fix $t_2 \notin B^1_x$. Consider,
\begin{align*}
\lefteqn{\expct_{t = (t_1,t_2,\{i\}) \leftarrow T_{t_2}}  \Pr[A=1 | ~  X_i = 0 , T=t]  \Pr[B=1 | ~  Y_i = 0 , T=t]  Bad^2_x } \\
& =  c(t_2) \expct_{t = (t_1,t_2,\{i\}) \leftarrow T_{t_2}}  \Pr[A=1 | ~  X_i = 0 , T=t] Bad^2_x  \\
&  \leq 2  c(t_2) \expct_{t = (t_1,t_2,\{i\}) \leftarrow T_{t_2}}  \Pr[A=1 | ~   T=t] Bad^2_x  \\
&  =   2 c(t_2) r(t_2) \expct_{t = (t_1,t_2,\{i\}) \leftarrow T_{t_2}}   Bad^2_x  \\
&  \leq   \frac{2}{100} c(t_2) r(t_2)  \quad \mbox{(from Claim~\ref{claim:good})}\\
&  =    \frac{2}{100} c(t_2) \expct_{t = (t_1,t_2,\{i\}) \leftarrow T_{t_2}}  \Pr[A=1 | ~   T=t]   \\
&  =    \frac{2}{100} c(t_2) \expct_{t = (t_1,t_2,\{i\}) \leftarrow T_{t_2}}  \Pr[A=1 | ~ X_i=0,  T=t]   \\
&  =    \frac{2}{100}  \expct_{t = (t_1,t_2,\{i\}) \leftarrow T_{t_2}}  \Pr[A=1 | ~ X_i=0,  T=t] \Pr[B=1 | ~  Y_i = 0 , T=t]   .
\end{align*}
\end{proof}

We can now finally prove our lemma. Let $Bad = 1$ iff any of $Bad^1_x, Bad^1_y, Bad^2_x,  Bad^2_y$ is $1$, otherwise $0$.
\begin{align*}
\lefteqn{\Pr[A=B=1, \disj(XY) = 0 ]} \\
& =\frac{1}{4} \expct_{t = (t_1,t_2,\{i\}) \leftarrow T}  \Pr[A=1| ~ T=t, X_i=1]  \Pr[B=1| ~ T=t, Y_i=1] \quad \mbox{(from Claim~\ref{claim:good1})}\\
& \geq  \frac{1}{4} \expct_{t = (t_1,t_2,\{i\}) \leftarrow T}  \Pr[A=1| ~ T=t, X_i=1]  \Pr[B=1| ~ T=t, Y_i=1] (1- Bad) \\
& \geq  \frac{1}{36} \expct_{t = (t_1,t_2,\{i\}) \leftarrow T}  \Pr[A=1| ~ T=t, X_i=0]  \Pr[B=1| ~ T=t, Y_i=0] (1- Bad)  \\
& \geq  \frac{84}{3600} \expct_{t = (t_1,t_2,\{i\}) \leftarrow T}  \Pr[A=1| ~ T=t, X_i=0]  \Pr[B=1| ~ T=t, Y_i=0]  \quad  \mbox{(from Claim~\ref{claim:good2})} \\
& = \frac{7}{225} \Pr[A= B=1, \disj(XY) =1]  \quad  \mbox{(from Claim~\ref{claim:good1})} .
\end{align*}
This implies 
\begin{align*}
\Pr[\disj(X'Y') = 0 ] & = \Pr[\disj(XY) = 0 |~ A=B=1] \\
& = \frac{\Pr[\disj(XY) = 0 , A=B=1]}{\Pr[A=B=1]} \\
& \geq  \frac{7}{225} \cdot \frac{\Pr[\disj(XY) = 1 , A=B=1]}{\Pr[A=B=1]} \\
& =  \frac{7}{225} \cdot \Pr[\disj(X'Y') = 1]  \geq \frac{1}{70} .
\end{align*}
\end{proofof}

\bibliographystyle{alpha}

\bibliography{sdptwo}

\end{document}